\documentclass[aip,jmp,amsmath,amssymb]{revtex4-1}
\usepackage[latin9]{inputenc}
\usepackage{amsmath}
\usepackage{amssymb}

\makeatletter

\providecommand{\tabularnewline}{\\}

 \@ifundefined{textcolor}{}
 {%
   \definecolor{BLACK}{gray}{0}
   \definecolor{WHITE}{gray}{1}
   \definecolor{RED}{rgb}{1,0,0}
   \definecolor{GREEN}{rgb}{0,1,0}
   \definecolor{BLUE}{rgb}{0,0,1}
   \definecolor{CYAN}{cmyk}{1,0,0,0}
   \definecolor{MAGENTA}{cmyk}{0,1,0,0}
   \definecolor{YELLOW}{cmyk}{0,0,1,0}
 }

\usepackage{amsmath}
\usepackage{amsthm}

\usepackage{graphicx}% Include figure files
\usepackage{dcolumn}% Align table columns on decimal point
\usepackage{bm}% bold math
\newtheorem{theorem}{Theorem}[section]
\newtheorem{lemma}[theorem]{Lemma}
\newtheorem{corollary}[theorem]{Corollary}
\newtheorem{definition}{Definition}[section]
\newtheorem{remark}[theorem]{Remark}
\newtheorem{example}[theorem]{Example}

\makeatother

\begin{document}

\title{A law of order estimation and leading-order terms for a family of
averaged quantities on a multibaker chain system}

\author{Hideshi Ishida}

\email{ishida@me.es.osaka-u.ac.jp}

\affiliation{Department of Mechanical Science and Bioengineering, 1-3 Machikaneyama,
Toyonaka, Osaka 560-8531, Japan}
\begin{abstract}
In this study a family of local quantities defined on each partition
and its averaging on a macroscopic small region, site, are defined
on a multibaker chain system. On its averaged quantities, a law of
order estimation (LOE) in the bulk system is proved, making it possible
to estimate the order of the quantities with respect to the representative
partition scale parameter $\Delta$. Moreover, the form of the leading-order
terms of the averaged quantities is obtained, and the form enables
us to have the macroscopic quantity in the continuum limit, as $\Delta\rightarrow0$,
and to confirm its partitioning independency. These deliverables fully
explain the numerical results obtained by Ishida, consistent with
the irreversible thermodynamics.
\end{abstract}
\maketitle

\section{Introduction}

It has been widely accepted that a coarse-grained entropy ($\varepsilon$-entropy)
should be employed in order for expanding the nonequilibrium statistical
mechanics based on the Gibbs entropy \cite{gilbert2,vollmer1,vollmer5,vollmer3,vollmer2,ishida1,tasaki1,gilbert3,gilbert1,gaspard2,dorfman2,barra2,barra3,cross1,cross2,cross3,gaspard1,vollmers}.
The coarse graining or the partitioning of the phase space, often
accomplished by a Markov partition, encounters the problem to find
a local or coarse-grained form that recovers the nonequilibrium thermodynamics
in the macroscopic or large system limit. 

The change of the coarse-grained entropy on a macroscopic small region,
e.g. site on a multibaker chain, has been decomposed into three terms:
entropy flow, entropy flow due to a thermostat, and nonnegative entropy
production \cite{gilbert2}, and its similar or extended formalism
is introduced by Vollmer, Breymann, M\'{a}ty\'{a}s, et al. \cite{vollmer1,vollmer5,vollmer3,vollmers,cross1,cross2,cross3}.
Although the entropy balance equation, based on the information theory,
is confirmed to agree with the phenomenological one in the macroscopic
limit, each term does not necessarily agree with a corresponding macroscopic
one: the estimated entropy flow due to a thermostat is contaminated
by a flux term and the contributions of convection and diffusion cannot
be distinguished \cite{ishida2012}. The conventional decomposition
inherently suffers from the need of an appropriate physical principle
when finding a coarse-grained form corresponding to an arbitrary macroscopic
quantity. 

On the other hand, Ishida \cite{ishida1,ishida2012} introduced a
completely different way to find local or coarse-grained terms. It
can be interpreted as the decomposition, not of the quantity of the
coarse-grained entropy change on a macroscopic small region, but of
the time-evolution operator of local entropy density defined on each
partition. In this approach, the decomposition is first performed
on the level of a partition or a master equation, based on the symmetric
properties of macroscopic entropy balance equation, and the local
terms are identified. The properties are symmetry or anti-symmetry
for the inversion of partition, density pairs and a given drift velocity.
Next, a spatio-temporal averaging of the local quantities are performed
on a macroscopic small region. The essence of the coarse graining
of this method lies in the averaging, different from the conventional
method putting emphasis on the partitioning of the phase space. In
this approach the quantity defined on each partition is, therefore,
appropriate to be called not coarse-grained quantity but local one.
Finally, the macroscopic (continuum) limit, where the representative
partition scale $\Delta$ goes to zero, are taken to compare with
macroscopic quantities. On a volume-preserving \cite{ishida1} and
a dissipative \cite{ishida2012} multibaker chain systems, each term
of the local entropy balance equation recovers the corresponding phenomenological
term of irreversible thermodynamics. The mathematical procedure which
stands on the symmetry, in principle, is applicable to find a local,
partition-level form of an arbitrary macroscopic quantity, or to find
the local balance equation corresponding to a macroscopic one for
an arbitrary quantity.

The procedure makes a family of local quantity which is symmetric
or anti-symmetric for some inversions and needs the property for the
limiting behavior of its averaged quantity. It is quantified by the
order exponent \textit{q} of a given averaged quantity which behaves
like $O(\Delta^{q})$ in the macroscopic limit. The order explains
the macroscopic observability and the scale separation of vanishing
quantities: when \textit{q} equals zero, the averaged quantity of
a given local form has a limit value and it is macroscopically observable.
In contrast, when \textit{q} is positive the averaged quantity vanishes
and it is not macroscopically observable. For sufficiently large positive
\textit{q}, the quantity approaches more rapidly to zero in the limit,
and therefore the scale of such a quantity is well separated from
the observable one. Moreover, Ishida \cite{ishida1} showed that the
positivity of entropy production in the volume-preserving system relies
on the property of a local residual term and numerically confirmed
that the order exponent of its averaged quantity is positive. It is
natural for us to consider that an isolated system is governed by
a volume-preserving dynamical system, and therefore, the behavior
is responsible for the law of increasing entropy. Thus, the exponent
\textit{q} is physically essential.

In this study, a family of local quantities defined on each partition
and its averaging are first defined (Sec. \ref{sec:Averaging}). Then,
we shall prove the law of order estimation (LOE) for the averaged
values of the local quantities on the multibaker chain system (Sec.
\ref{sec:LOE}). It is applicable to a family of local quantity in
the bulk system, and all of the important orders mentioned above can
be estimated. The estimation law is a slight extension originally
conjectured by Ishida \cite{ishida2012} from the numerical experiments
on two multibaker chain systems\cite{ishida1,ishida2012}. Finally,
we shall obtain the form of the leading-order terms of the averaged
quantities in Sec. \ref{sec:leading-term}. It is quite important
because it is just the form of the macroscopic quantity when \textit{q}
equals zero. In the coarse graining approach it is expected that these
results are independent of partitioning, and the leading form also
enables us to discuss the independency that has been numerically confirmed
on the multibaker chain \cite{ishida1,ishida2012}.

\section{Averaging of a family of local quantities\label{sec:Averaging}}

Ishida \cite{ishida2012} introduced a family of local quantities
defined on each partition. They are found by the decomposition of
the equations for local probability density and local entropy density,
based on symmetric properties. When the local quantities is averaged
on a macroscopic small region, Ishida numerically confirmed that the
two decomposed equations recover the corresponding macroscopic balance
equations of the irreversible thermodynamics.

Now let us consider a slightly extended family of local quantities,
called \textit{T}-type local quantities and its averaged value. In
what follows, we utilize the notation of Ishida \cite{ishida2012}.

\begin{definition}[\textit{T}-type local quantity]\label{def:TLQ}
Let \textit{i},\textit{j} be integers to identify the partition number
of a Markov partition, and suppose that the transition volume $\tilde{W}_{ji}$
or $\tilde{W}_{ij}$ is nonzero, where $\tilde{W}_{ji}\equiv W_{ji}\Delta V_{i}$
by use of the transition probability $W_{ji}$ from \textit{i}-th
to \textit{j}-th partition and the partition volume (Liouville measure)
$\Delta V_{i}$ of the \textit{i}-th partition. Let $e_{i}$, defined
by $\Delta V_{i}/\sum_{k}\tilde{W}_{ik}$, denote the expansion rate
on the \textit{i}-th partition and $\rho_{i}^{(l)}$, defined by the
probability measure $P_{i}^{(l)}$ on the \textit{i}-th partion at
the \textit{l}-th step over the partition volume $\Delta V_{i}$,
denote the local density. Suppose also that integers $\epsilon$ and
$\delta$ are either 1 or -1 and that $f(x,y)$ is symmetric or antisymmetric
for the interchange of \textit{x} and \textit{y}. Then we define the
following family $A_{i}^{(l)}$ of local quantity on the \textit{i}-th
(arbitrary) partition at the \textit{l}-th time step, called a transitional
type (\textit{T}-type) local quantity:

\begin{subequations}\label{eq:Ai+Aij}

\begin{eqnarray}
A_{i}^{(l)} & \equiv & \sum_{j}A_{ij}^{(l)},\label{eq:Ai}\\
A_{ij}^{(l)} & \equiv & \left[(1+\delta e_{j})\tilde{W}_{ji}+\epsilon(1+\delta e_{i})\tilde{W}_{ij}\right]f(\rho_{j}^{(l)},\rho_{i}^{(l)}).\label{eq:Aij}
\end{eqnarray}
\end{subequations} \end{definition}

\begin{remark}The summand $A_{ij}^{(l)}$ is symmetric or antisymmetric
for the interchange of partition pair (i,j) and for the inversion
of external parameter, such as the drift velocity \textit{v}, for,
at least, the multibaker chain system. The quantity is transitional
because it is made by the decomposition of a transport equation associated
with a master equation. As an arbitrary function \textit{g}(\textit{x},\textit{y})
consists of symmetric and antisymmetric parts, the local quantity
of the form $\sum_{j}\tilde{W}_{ij}g(\rho_{j}^{(l)},\rho_{i}^{(l)})$
can be expressed as the summation of \textit{T}-type local quantities.\end{remark}

\begin{example}Set $\delta=\epsilon=-1$ and $f=(\rho_{j}^{(n)}+\rho_{i}^{(n)})/4$.
Then
\begin{equation}
A_{ij}^{(l)}=\frac{\tilde{U}_{j\leftarrow i}^{[a]}}{\Delta_{i,j}}\frac{\rho_{j}^{(l)}+\rho_{i}^{(l)}}{2},\,\tilde{U}_{j\leftarrow i}^{[a]}\equiv\frac{(1-e_{j})\tilde{W}_{ji}-(1-e_{i})\tilde{W}_{ij}}{2}\Delta_{i,j},\label{eq:exam1}
\end{equation}
where $\Delta_{i,j}$ is expected to reflect the distance between
\textit{i}-th and \textit{j}-th partitions. This is a component of
local probability density change\cite{ishida2012}.\end{example}

\begin{example}\label{exm:rs}The residual entropy source $r_{s,i}^{(l)}$
(Eq. (20f) of Ref. \onlinecite{ishida2012}) can be expressed by the
summation of some\textit{ T}-type quantities. \end{example}

Hereafter, we also assume that the function \textit{f} is of class
$C^{\infty}$ and has a nonnegative integer \textit{m}, called a characteristic
exponent, and a characteristic function $\hat{C}(x)$, defined as
follows.

\begin{definition}[Characteristic exponent and characteristic function]

If there exists the following power index \textit{m} and nonzero bounded
function $ $$\hat{C}(x)$:
\begin{equation}
^{\forall}s,t(\neq s)\in\mathbb{R},\,^{\exists}m\in\mathbb{R},\,\lim_{h\rightarrow0}\frac{f(x+sh,x+th)}{\left(\frac{s-t}{2}h\right)^{m}}=\hat{C}(x),\label{eq:defChat}
\end{equation}
then the exponent \textit{m} and the function $ $$\hat{C}(x)$ are
called a characteristic exponent and a characteristic function, respectively.

\end{definition}

\begin{remark}

Assume that a function \textit{f} of the variable \textit{$X(\equiv(x-y)/2)$}
and \textit{$Y(\equiv(x+y)/2)$} is of class $C^{1}$ in the neighborhood
of $X=0$. Then the characteristic exponent \textit{m} can be expressed
as:

\[
m=\lim_{X\rightarrow0}\frac{X}{f}\frac{\partial f(X,Y)}{\partial X}.
\]

\end{remark}

\begin{definition}[$T^1$-type local quantity]

Assume that a function $f(x,y)$ to define a \textit{T}-type local
quantity is of class $C^{\infty}$ and that \textit{f} has a constant,
nonnegative integer characteristic exponent \textit{m} and a characteristic
function $ $$\hat{C}(x)$. Then this type of local quantity $A_{i}^{(l)}$
is called $T^{1}$-type, expressed by $A^{(l)}(f,m,\hat{C})$.\end{definition}

\begin{remark}It imply that the function \textit{f} can be expressed
as
\begin{eqnarray}
f(x,y) & = & \tilde{f}(X,Y)\equiv\tilde{f}(\frac{x-y}{2},\frac{x+y}{2})\nonumber \\
 & = & \hat{C}(Y)X^{m}+C^{(m+2)}(Y)X^{m+2}+C^{(m+4)}(Y)X^{m+4}+\cdots,\label{eq:expansion_f}
\end{eqnarray}
in a neighborhood of $X=0$.\end{remark}

\begin{remark}When the absolute value of \textit{X} is sufficiently
small, \textit{f} can be well approximated by $ $$\hat{C}(Y)X^{m}$.
Replacing \textit{f} in Eq. (\ref{eq:Aij}) with $\hat{C}(Y)X^{m}$
and comparing the local quantity $A_{ij}^{(n)}$ with Table 1 in Ref.
\onlinecite{ishida2012}, we are easily convinced that the quantity
of the present study is an extension in the sense that the symmetric
part $\hat{C}(Y)\,(Y\equiv(\rho_{j}^{(l)}+\rho_{i}^{(l)})/2)$ is
not restricted to the product of the power of two symmetric order
components $Y$ and $\ln Y$.\end{remark}

For the above family of local quantity $A_{i}^{(l)}$, Eq. (\ref{eq:Aij}),
we can introduce its averaged quantity $\left\langle A\right\rangle _{R}^{(l)}$
on a macroscopic small region \textit{R} \cite{ishida2012}.

\begin{definition}[Spatio-temporal averaging on $R$]

Let \textit{R} denote a union of partitions of a Markov partition,
$\Delta V_{R}$ the volume of the region \textit{R}, $\tau$ timely
increment (time step) of a given map. For a local quantity $A_{i}^{(l)}$,
defined on the \textit{i}-th partition at the \textit{l}-th step,
we can define its spatio-temporal averaging on \textit{R} as follows:

\begin{equation}
\left\langle A\right\rangle _{R}^{(l)}=\frac{1}{\tau\Delta V_{R}}\sum_{i\subset R}A_{i}^{(l)},\label{eq:def_<A>^(l)_R}
\end{equation}
where $i\subset R$ means that the quantity $A_{i}^{(l)}$ is summed
with respect to all the partitions included in the region \textit{R}.
\end{definition}

In the next section we see that a law of order estimation (LOE) holds
for the averaged quantity in the bulk of a multibaker chain system.

\section{Law of order estimation (LOE)\label{sec:LOE}}

Hereafter we shall confine our discussion to a triadic multibaker
chain, introduced by Vollmer et al. \cite{vollmer1,vollmer2,vollmer3,vollmer5}
and Ishida \cite{ishida1,ishida2012}. 

\begin{definition}[Multibaker-chain map]

Let $R_{m}\equiv[0,N]\times[0,\Delta r]\times[0,h]$ be a domain defined
on $\mathbb{Z\times\mathbb{R^{\textrm{2}}}}$. Consider the following
map $T:\, R_{m}\rightarrow R_{m}$, called a multibaker-chain map.
For the case of $0<n<N$,

\begin{align}
 & T(n,x,y)=\nonumber \\
 & \begin{cases}
(n-1,x/\eta_{n}(0),\nu_{n-1}(0)y), & 0\leq x<\eta_{n}(0)\Delta r,\\
(n,(x-\eta_{n}(0)\Delta r)/\eta_{n}(1),\nu_{n}(0)h+\nu_{n}(1)y), & \eta_{n}(0)\Delta r\leq x<(\eta_{n}(0)\\
 & +\eta_{n}(1))\Delta r,\\
(n+1,[x-(\eta_{n}(0)+\eta_{n}(1))\Delta r]/\eta_{n}(2), & (\eta_{n}(0)+\eta_{n}(1))\Delta r\\
(\nu_{n+1}(0)+\nu_{n+1}(1))h+\nu_{n+1}(2)y), & \leq x\leq\Delta r,
\end{cases}\label{eq:defT}
\end{align}
and an appropriate map is given for the case of $n=0$ or \textit{N}.
Herein (\textit{n},\textit{x},\textit{y}) is regarded as a position
(\textit{x},\textit{y}) on the \textit{n}th site, $A_{n}\subset\mathbb{R^{\textrm{2}}}$,
and $A_{0}$ and $A_{N}$ is regarded as two boundary sites. The invertible
condition\cite{vollmer1,vollmer5,ishida2012} yields $\nu_{n}(\omega)=\eta_{n}(2-\omega)$
.

\end{definition}

\begin{remark}

By convention, the domain\textit{ $R_{m}$} is regarded as a one-dimensional
chain of (\textit{N}+1) rectangular sites, and therefore, the map
is called a multibaker chain system. Each site has the origin fixed
at the lower left corner and the coordinates \textit{x} and \textit{y}
are the horizontal and vertical direction, respectively. We shall
take \textit{r} direction in which two disjoint sites connect, and
the width $\Delta r$ is treated as the distance between the two sites.\end{remark}

In what follows, a restricted type of multibaker chain system is dealt
with.

\begin{definition}[UOFP-type multibaker-chain map]\label{def:UOFP}

Consider the following transient probability $\eta_{n}(\omega)$ on
each site $A_{n}$ that is uniform except the two boundary sites.

\begin{equation}
\eta_{n}(0)=\frac{\beta}{2}(1-\frac{Pe_{g}}{2}),\eta_{n}(2)=\frac{\beta}{2}(1+\frac{Pe_{g}}{2}),\eta_{n}(1)=1-\beta,\label{eq:eta-nu}
\end{equation}
where $\beta\equiv2\tau D/\Delta r^{2}$, and \textit{D} denotes a
diffusion coefficient, $\tau$ a time step. $Pe_{g}(\equiv v\Delta r/D)$
is the so-called grid P\'{e}clet number \cite{patankar}. The condition
that $\eta(\omega)$ should be positive leads to the order of the
time step $\tau$ must be greater than or equal to the order of $\Delta r^{2}$,
and Ishida \cite{ishida1,ishida2012} fixed the order at $\Delta r^{2}$.
In other words, the order exponent of $\beta$ with respect to $\Delta r$
must be nonnegative. At the boundary sites \textit{n}=0 and \textit{N},
other transient probabilities are given, determining the macroscopic
boundary condition. Once all of the transition probabilities are given,
a Markov partition $(n,\underline{\omega}_{k})$ on the \textit{n}th
site, the partition volume $\Delta V(n,\underline{\omega}_{k})$ on
the partition, and the probability measure $P^{(l)}(n,\underline{\omega}_{k})$
at the \textit{l}th step are defined, where $\underline{\omega}_{k}(\equiv\omega_{0}\omega_{1}\cdots\omega_{k-1})$
is a \textit{k}-digit trit number specifying a partition of width
$\Delta r$ by height $h\times\nu_{n}(\underline{\omega}_{d})$, forming
a line in its numerical order from bottom to top on each site \cite{ishida2012}.
The ``bulk'' of this system is defined by the sites $A_{n}$ for
$0<n-k<n<n+k<N$ and time step $l>k$.

In this system, the limit of the partition scale $\Delta\equiv\Delta r\rightarrow0$
with the total length \textit{$L\equiv N\Delta r$} and trit number
\textit{k} fixed defines a continuum limit. In this limit the averaged
equation for $P^{(l)}(n,\underline{\omega}_{k})$ on a bulk site $A_{n}$
recovers the one-dimensional Fokker-Planck equation for macroscopic
(averaged) measure density $\rho$ \cite{ishida2012}. In addition,
the distance between the bulk domain composed by all bulk sites and
the boundary sites approaches to zero. Hereafter, the multibaker-chain
map (system) mentioned above is called a uniform, one-dimensional
Fokker-Planck (UOFP) type .\end{definition} 

\begin{remark}The type is independent of the dynamics on the boundary
sites.\end{remark}

The definition of the bulk system for the UOFP-type multibaker chain
allows us to express the local density on each partition by
\begin{eqnarray}
\rho^{(l)}(n,\underline{\omega}_{k}) & \equiv & P^{(l)}(n,\underline{\omega}_{k})/\Delta V(n,\underline{\omega}_{k})\nonumber \\
 & = & \frac{\eta_{n+1-\omega_{0}}(\omega_{0})P^{(l-1)}(n+1-\omega_{0},\stackrel{\leftarrow}{\underline{\omega}}_{k-1})}{\nu_{n}(\omega_{0})\Delta V(n+1-\omega_{0},\stackrel{\leftarrow}{\underline{\omega}}_{k-1})}\nonumber \\
 & = & \frac{\eta(\omega_{0})}{\nu(\omega_{0})}\rho^{(l-1)}(n+1-\omega_{0},\stackrel{\leftarrow}{\underline{\omega}}_{k-1}),\label{eq:rho(n,wk)}
\end{eqnarray}
where $\stackrel{\leftarrow}{\underline{\omega}}_{k-1}\equiv\omega_{1}\cdots\omega_{k-1}$
and this is the so-called left shift of symbolic dynamics. Please
note herein that the volume expansion rate $e(n,\underline{\omega}_{k})$,
i.e. the reciprocal of volume contraction rate, at a partition $(n,\underline{\omega}_{k})$
can be written as
\begin{eqnarray}
e(n,\underline{\omega}_{k}) & = & \Delta V(n,\underline{\omega}_{k})/\left[\eta_{n+1-\omega_{0}}(\omega_{0})\Delta V(n+1-\omega_{0},\underline{\overleftarrow{\omega}}_{k-1})\right]\nonumber \\
 & = & \nu_{n}(\omega_{0})/\eta_{n+1-\omega_{0}}(\omega_{0})\label{eq:e(n,wk)}
\end{eqnarray}
from the definition of expansion rate (Def. \ref{def:TLQ}). It is
essential that the relation (\ref{eq:rho(n,wk)}) holds recursively
and the product of the contraction rates is independent of the site
number \textit{n} while the \textit{n}th site considered is in the
bulk. Conversely, if the \textit{n}th site is outside the bulk, the
product must be contaminated by the discontinuous contraction rate
given at the boundary sites, and Ishida \cite{ishida2012} has numerically
confirmed that the LOE, proved below, is not fulfilled for such non-bulk
sites near the boundaries. This is the onset of boundary effects \cite{cross1,ishida2012}.
In what follows, the subscripts of $\eta$ and $\nu$ are omitted
because these transient probabilities are independent of the site
number \textit{n} in the bulk system.

Substituting Eqs. (\ref{eq:rho(n,wk)}) and (\ref{eq:e(n,wk)}) into
Eq. (\ref{eq:Aij}), we obtain the following form of the averaged
quantity $\left\langle A\right\rangle _{R}^{(l)}$, Eq. (\ref{eq:def_<A>^(l)_R}),
on a bulk site $R=A_{n}$ for a $T^{1}$-type local quantity $A^{(l)}(f,m,\hat{C})$
as follows:

\begin{subequations}\label{eq:cf_<A>^(l)_An}
\begin{equation}
\left\langle A\right\rangle _{A_{n}}^{(l)}\text{\ensuremath{\equiv}}K/(\tau\Delta r),\label{eq:<A>^(l)_n}
\end{equation}
where
\begin{eqnarray}
K & = & \sum_{\omega}\sum_{\omega'}\sum_{\underline{\omega}_{k}}\left[\left(\eta(\omega')+\delta\nu(\omega')\right)\nu(\underline{\omega}_{k})\Delta r\right.\nonumber \\
 &  & \times f\left(\rho^{(l)}(n+\omega'-1,\omega'\underline{\omega}_{k-1}),\,\rho^{(l)}(n,\underline{\omega}_{k})\right)\nonumber \\
 & + & \epsilon\left(\eta(\omega_{0})+\delta\nu(\omega_{0})\right)\nu(\underline{\overleftarrow{\omega}}_{k-1}\omega)\Delta r\nonumber \\
 &  & \left.\times f\left(\rho^{(l)}(n+1-\omega_{0},\underline{\overleftarrow{\omega}}_{k-1}\omega),\,\rho^{(l)}(n,\underline{\omega}_{k})\right)\right]\nonumber \\
 & = & \sum_{\omega}\sum_{\underline{\omega}_{k}}\left(\eta(\omega)+\delta\nu(\omega)\right)\nu(\underline{\omega}_{k})\mathcal{F}_{n}^{(l)}(\epsilon,m,\omega,\underline{\omega}_{k})\Delta r,\label{eq:K}
\end{eqnarray}
and
\begin{eqnarray}
\mathcal{F}_{n}^{(l)}(\epsilon,m,\omega,\underline{\omega}_{k}) & \equiv & F_{n}^{(l)}(\omega,\underline{\omega}_{k})+\epsilon(-1)^{m}F_{n+1-\omega}^{(l)}(\omega,\underline{\omega}_{k}),\label{eq:defcalF^(l)_n}\\
F_{n}^{(l)}(\omega,\underline{\omega}_{k}) & \equiv & f\left(\frac{\eta(\omega)}{\nu(\omega)}\rho^{(l-1)}(n,\stackrel{\leftarrow}{\underline{\omega}}_{k-1})\right.\nonumber \\
 &  & \left.,\frac{\eta(\omega_{0})}{\nu(\omega_{0})}\rho^{(l-1)}(n+1-\omega_{0},\stackrel{\leftarrow}{\underline{\omega}}_{k-1})\right).\label{eq:F}
\end{eqnarray}
\end{subequations}

The following lemma plays essential role for estimating the order
of \textit{K} or $\left\langle A\right\rangle _{A_{n}}^{(l)}$.

\begin{lemma}\label{lem:order_K} The order exponent of $K/(\beta\Delta r)$
for $m\neq0$ or $K/\Delta r$ for $m=0$ with respect to $\Delta r$
is even number, determined by the property of the summation with respect
to only two trinary bits $\omega$ and $\omega_{0}$.\end{lemma}

\begin{proof}Firstly, we define the following function of an integer
\textit{a}: 
\begin{eqnarray}
S_{\nu}(a) & \equiv & \begin{cases}
1, & a=0,\\
\frac{\beta Pe_{g}}{2}, & a:odd,\\
\beta, & a(\neq0):even,
\end{cases}\label{eq:S_nu(a)}
\end{eqnarray}
and it is originated from the quantity $\sum_{\omega}\nu(\omega)I_{\omega}^{a}$,
where $I_{\text{\ensuremath{\omega}}}\text{\ensuremath{\equiv}1--}\omega$.
For $a\neq0$, the definition of bulk transient probability (\ref{eq:eta-nu})
and the invertible condition (Def. \ref{def:TLQ}) have the quantity
agree with $S_{\nu}(a)$. The definition at $a=0$ is useful to express
the summation with respect to $\omega$ when the summand does not
include the form $I_{\omega}^{a}$ of $a\neq0$. Herein, it is important
to note that the function increases its order by one with respect
to $\Delta$ when \textit{a} changes from an even number to an odd
one because the order of $Pe_{g}$ is $O(\Delta)$. In addition, we
define the following function related to the above function 
\begin{equation}
S_{\eta\nu}(\delta,a)\equiv(\delta+(-1)^{a})S_{\nu}(a)\sim\sum_{\omega}(\eta(\omega)+\delta\nu(\omega))I_{\omega}^{a},\label{eq:S_etanu}
\end{equation}
and, therefore, the function inherit the order-increase property.

On the other hand, when the partition scale $\Delta r$ is sufficiently
small we have the following relations:\begin{subequations}\label{eq:sum-operations}
\begin{equation}
\frac{\eta(\omega)}{\nu(\omega)}=1+\sum_{n=1}^{\infty}(-1)^{n}\frac{(Pe_{g}I_{\omega})^{n}}{2^{n-1}},\label{eq:etabynu}
\end{equation}
\begin{eqnarray}
\rho^{(l-1)}(n+1-\omega_{0},\stackrel{\leftarrow}{\underline{\omega}}_{k-1}) & = & \rho^{(l-1)}(n,\stackrel{\leftarrow}{\underline{\omega}}_{k-1})+\left(\frac{\partial\rho}{\partial r}\right)^{(l-1)}(n,\stackrel{\leftarrow}{\underline{\omega}}_{k-1})(I_{\omega_{0}}\Delta r)^{1}\nonumber \\
 & + & O^{(l-1)}\left((I_{\omega_{0}}\Delta r)^{2}\right),\label{eq:rho(n+1-omg)}
\end{eqnarray}
\begin{equation}
\sum_{\stackrel{\leftarrow}{\underline{\omega}}_{k-1}}\nu(\stackrel{\leftarrow}{\underline{\omega}}_{k-1})g\left(A_{n}^{(l-1)}(\stackrel{\leftarrow}{\underline{\omega}}_{k-1})\right)=g(A_{n}^{(l)})+O(\Delta r),\label{eq:sum_nugA}
\end{equation}
\end{subequations}where $A_{n}^{(l)}(\underline{\omega}_{k})$ is
a family of quantities that fulfills recursively the following relation
\begin{equation}
A_{n}^{(l)}(\underline{\omega}_{k})=\frac{\eta(\omega_{0})}{\nu(\omega_{0})}A_{n+1-\omega_{0}}^{(l-1)}(\underline{\overleftarrow{\omega}}_{k-1}),\label{eq:defA^(l)_n}
\end{equation}
such as $\rho^{(l-1)}(n,\stackrel{\leftarrow}{\underline{\omega}}_{k-1})$
and $\left(\partial\rho/\partial r\right)^{(l-1)}(n,\stackrel{\leftarrow}{\underline{\omega}}_{k-1})$.
Herein, $A_{n}^{(l)}$, defined by $A_{n}^{(l)}(\underline{\omega}_{0})$,
is a family of quantities given on the \textit{n}th site. 

It is also worth noting that the scale parameter $\Delta r$ or $Pe_{g}$
in Eq. (\ref{eq:sum-operations}) is accompanied by $I_{\omega}$
or $I_{\omega_{0}}$. It follows that odd-ordered terms with respect
to $\Delta r$ in the summand of \textit{$\mathsf{\mathbb{\mathfrak{\mathcal{F}}}}_{n}^{(l)}$,}
Eq. (\ref{eq:defcalF^(l)_n}), increases their order by one through
the operations (\ref{eq:S_nu(a)}) and (\ref{eq:S_etanu}). Consequently,
the order exponent of $K/(\beta\Delta r)$ for the case of $m\neq0$
or $K/\Delta r$ for $m=0$ must be even number. Note herein that
the coefficient $\beta$ does not appear for the case of $m=0$ because
$S_{\nu}(0)=1$. In addition, Eq. (\ref{eq:sum_nugA}) shows that
the summation with respect to $\stackrel{\leftarrow}{\underline{\omega}}_{k-1}$
is order one ($O(\Delta^{0}))$ and does not affect the order of \textit{K}.
That's why we can estimate the order of \textit{K }or\textit{ $\left\langle A\right\rangle _{A_{n}}^{(l)}$}
from the properties of the summation with respect to only two trinary
bits $\omega$ and $\omega_{0}$. \end{proof}

\begin{remark}Eqs. (\ref{eq:defChat}), (\ref{eq:sum-operations})
and (\ref{eq:defA^(l)_n}) show that the leading-order terms of $F_{n}^{(l)}(\omega,\underline{\omega}_{k})$,
Eq. (\ref{eq:F}), is $O(\Delta^{m})$ and that the leading terms
are identical to that of $F_{n+1-\omega}^{(l)}(\omega,\underline{\omega}_{k})$.
Thus the leading-order term of \textit{$\mathsf{\mathbb{\mathfrak{\mathcal{F}}}}_{n}^{(l)}$,}
Eq. (\ref{eq:defcalF^(l)_n}), has the coefficient of the quantity
$1+\epsilon(-1)^{m}$.\end{remark}

\begin{remark}In Eq. (\ref{eq:rho(n+1-omg)}), the term $\left(\partial\rho/\partial r\right)^{(l-1)}(n,\stackrel{\leftarrow}{\underline{\omega}}_{k-1})$
is merely a formal expression based on the Taylor expansion. From
the recursive definition (\ref{eq:def_<A>^(l)_R}), it is defined
if and only if the derivative of $\rho_{n}^{(l)}(\underline{\omega}_{0})$
with respect to \textit{r} at $r=n\Delta r$ is defined in the continuum
limit, $\Delta r\rightarrow0$. The property fully depends on the
appropriateness of both the macroscopic small region \textit{$A_{n}$}
and the transient probabilities.\end{remark}

Now we are in a position to prove the following law of order estimation
that enable us to estimate the order of the averaged quantity $\left\langle A\right\rangle _{A_{n}}^{(l)}$. 

\begin{theorem}[\bf{LOE: Law of order estimation}]\label{theorem:LOE}Consider
the averaged value $\left\langle A\right\rangle _{A_{n}}^{(l)}$ of
a $T^{1}$-type local quantity $A^{(l)}(f,m,\hat{C})$ for a UOFP-type
multibaker chain system partitioned by a Markov partition $(n,\underline{\omega}_{k})$.
The order of $\left\langle A\right\rangle _{A_{n}}^{(l)}$ can be
expressed as follows:

\begin{description}

\item[(Case I)]

\begin{subequations}\label{eq:LOE}

If the characteristic exponent \textit{m} is an even number, then
\begin{equation}
\left\langle A\right\rangle _{A_{n}}^{(l)}=\frac{K}{\tau\Delta r}\sim\begin{cases}
O(\tau^{-1}), & m=0\textrm{ and }\epsilon=\delta=1,\\
O(\Delta r^{m-2}), & m>0\textrm{ and }\epsilon=\delta=1,\\
O(\Delta r^{m}), & \textrm{elsewhere.}
\end{cases}\label{eq:LOE_caseI}
\end{equation}

\item[(Case II)]If \textit{m} is an odd number, then
\begin{equation}
\left\langle A\right\rangle _{A_{n}}^{(l)}=\frac{K}{\tau\Delta r}\sim\begin{cases}
O(\Delta r^{m-1}), & \epsilon=-1\textrm{ or }\delta=1,\\
O(\Delta r^{m+1}), & \textrm{elsewhere}.
\end{cases}\label{eq:LOE_caseII}
\end{equation}

\end{subequations}\end{description}\end{theorem}

\begin{proof}

\subsubsection*{(Case I) }

If \textit{f} is symmetric for the interchange of \textit{x} and \textit{y},
i.e. \textit{m} is an even number, the order of \textit{K} depends
on $\epsilon$ and $\delta$, refined two cases: (I-1) $\epsilon=\delta=1$
and (I-2) elsewhere. For the former case, $1+\epsilon(-1)^{m}\neq0$
and $\delta+1\neq0$. From Lemma \ref{lem:order_K}, therefore, \textit{m}th
order terms in $F_{n}^{(l)}(\omega,\underline{\omega}_{k})$ with
the coefficient of $I_{\omega}^{2a}I_{\omega_{0}}^{2b}$, such that
$2(a+b)=m$ for two integers \textit{a} and \textit{b}, determines
the leading order of \textit{K}. Considering the operations (\ref{eq:S_nu(a)})
and (\ref{eq:S_etanu}), we obtain
\[
K\sim\begin{cases}
\Delta r, & m=0,\\
\Delta r\beta\Delta r^{m}, & m>0.
\end{cases}
\]
For the case of $m\neq0$, \textit{K} can have some terms with the
coefficient of $\beta^{2}$ for the case of $a\neq0$ and $b\neq0$.
Note that, however, such terms cannot be, in general, the leading-order
terms because $\beta$ has the nonnegative order exponent (cf. Def.
\ref{def:UOFP}). They are the leading terms if and only if $\beta$
is order one. That is why \textit{K} can be expressed as above mentioned.

For the latter case (I-2), the \textit{m}th order terms vanish. The
next (\textit{m}+1)th order terms have the coefficient of $I_{\omega}^{a}I_{\omega_{0}}^{b}$
and one of the integer power indices \textit{a} and \textit{b} is
positive odd number, occurring the order increase with the coefficient
$\beta$ (Lemma \ref{lem:order_K}). This results in the same order
as the (\textit{m}+2)th order terms with the coefficient of $I_{\omega}^{2a}I_{\omega_{0}}^{2b}$,
such that $2(a+b)=m+2$. Therefore, 
\[
K\sim\Delta r\beta\Delta r^{m+2}.
\]

The order of averaged quantity $\left\langle A\right\rangle _{A_{n}}^{(l)}$,
Eq. (\ref{eq:<A>^(l)_n}), for the Case I, Eq. (\ref{eq:LOE_caseI}),
is estimated from these results.

\subsubsection*{(Case II)}

If \textit{f} is antisymmetric, i.e. \textit{m} is an odd number,
the order of \textit{K} also depends on $\epsilon$ and $\delta$,
refined two cases: (II-1) $\epsilon=-1\textrm{ or }\delta=1$ and
(II-2) elsewhere. For the former case, $1+\epsilon(-1)^{m}\neq0$
or $\delta+1\neq0$. When $1+\epsilon(-1)^{m}\neq0$, the \textit{m}th
order terms in $F_{n}^{(l)}(\omega,\underline{\omega}_{k})$ remains.
But the order increases by one through the operations (\ref{eq:S_nu(a)})
and (\ref{eq:S_etanu}). The leading terms have the coefficient of
$\beta$ because \textit{m} is positive. When $1+\epsilon(-1)^{m}=0$
and $\delta+1\neq0$, the \textit{m}th order terms vanish. But the
next (\textit{m}+1)th order terms have the coefficient of $I_{\omega}^{2a}I_{\omega_{0}}^{2b}$,
such that $2(a+b)=m+1$, remaining as the (\textit{m}+1)th order terms
with the coefficient $\beta(\delta+1)$. As a result, 
\[
K\sim\Delta r\beta\Delta r^{m+1}.
\]

For the latter case (II-2), the (\textit{m}+1)th order terms vanish.
The next (\textit{m}+2)th order terms have the coefficient of $I_{\omega}^{a}I_{\omega_{0}}^{b}$
and one of the integer power indices \textit{a} and \textit{b} is
positive odd number, occurring again the order increase with the coefficient
$\beta$. This results in the same order as the (\textit{m}+3)th order
terms with the coefficient of $I_{\omega}^{2a}I_{\omega_{0}}^{2b}$,
such that $2(a+b)=m+3$. For the case, therefore, 
\[
K\sim\Delta r\beta\Delta r^{m+3}.
\]

As a result, the order of averaged quantity $\left\langle A\right\rangle _{A_{n}}^{(l)}$
for the Case II is estimated as Eq. (\ref{eq:LOE_caseII}).\end{proof}

\begin{corollary}If the order of $\tau$ is $O(\Delta r^{2})$ and
the symmetric part $\hat{C}(Y)\,(Y\equiv(\rho_{j}^{(l)}+\rho_{i}^{(l)})/2)$
is restricted to the product of the power of two symmetric order components
$Y$ and $\ln Y$, the estimated orders (\ref{eq:LOE}) is the statement
of the LOE conjectured by Ishida \cite{ishida2012}. The LOE proved
here is an extension of the empirical LOE. \end{corollary}

\begin{remark}The resultant order exponent is always even number
with an exception for the case that $m=0$ and $\epsilon=\delta=1$.
Therefore, the vanishing-order terms are more than or equal to the
second order, and the property plays a role in well separating convergent
or macroscopic $O(1)$ terms from them. Such a behavior has been numerically
confirmed in residual terms, such as residual entropy source (see
below), residual entropy flux \cite{ishida1,ishida2012}.\end{remark}

\begin{example}For the case of the quantity (\ref{eq:exam1}), characteristic
exponent \textit{m}=0 but $\epsilon=\delta=-1$ and, therefore, its
averaged quantity is $O(1)$. This is a component of the macroscopic
density change \cite{ishida2012}.\end{example}

\begin{example}The averaged quantity of the residual entropy source
(Example \ref{exm:rs}) vanishes in the continuum limit because its
order is estimated to be $O(\Delta r^{2})$. This is responsible for
the positive entropy production in the volume-preserving system, i.e.
for the law of increasing entropy\cite{ishida1,ishida2012}. \end{example}

\section{Leading-order terms of an averaged quantity\label{sec:leading-term}}

The proof of the LOE, mentioned above, can be regarded as the introduction
for deriving the form of the leading-order terms in the bulk system.
Successively, we shall obtain the specific form of the leading terms
of the averaged value $\left\langle A\right\rangle _{A_{n}}^{(l)}$
of a $T^{1}$-type local quantity $A^{(l)}(f,m,\hat{C})$ for a UOFP-type
multibaker chain.

First of all, we shall define some variables and abridged notations
as follows:
\begin{eqnarray*}
x_{1} & \text{\ensuremath{\equiv}} & \frac{\eta(\omega)}{\nu(\omega)}\rho^{(l-1)}(n,\stackrel{\leftarrow}{\underline{\omega}}_{k-1}),\\
x_{2} & \equiv & \frac{\eta(\omega)}{\nu(\omega)}\rho^{(l-1)}(n+1-\omega,\stackrel{\leftarrow}{\underline{\omega}}_{k-1}),\\
y_{1} & \text{\ensuremath{\equiv}} & \frac{\eta(\omega_{0})}{\nu(\omega_{0})}\rho^{(l-1)}(n+1-\omega_{0},\stackrel{\leftarrow}{\underline{\omega}}_{k-1}),\\
y_{2} & \equiv & \frac{\eta(\omega_{0})}{\nu(\omega_{0})}\rho^{(l-1)}(n+1-\omega_{0}+1-\omega,\stackrel{\leftarrow}{\underline{\omega}}_{k-1}),
\end{eqnarray*}
and

\[
X_{i}\equiv\frac{x_{i}-y_{i}}{2},\, Y_{i}\equiv\frac{x_{i}+y_{i}}{2},\,(i=1\textrm{ or }2).
\]
In this section, we shall also utilize the abridged notation for the
quantity $\overleftarrow{A}_{n}^{(l)}\equiv A_{n}^{(l-1)}(\underline{\overleftarrow{\omega}}_{k-1})$,
defined in Eq. (\ref{eq:defA^(l)_n}). For example, $\rho^{(l-1)}(n,\stackrel{\leftarrow}{\underline{\omega}}_{k-1})$,
$\left(\partial\rho/\partial r\right)^{(l-1)}(n,\stackrel{\leftarrow}{\underline{\omega}}_{k-1})$
and $\left(\partial^{2}\rho/\partial r^{2}\right)^{(l-1)}(n,\stackrel{\leftarrow}{\underline{\omega}}_{k-1})$
are abridged to $\overleftarrow{\rho}_{n}^{(l)}$, $\overleftarrow{\left(\rho_{r}\right)}_{n}^{(l)}$
and $\overleftarrow{\left(\rho_{rr}\right)}_{n}^{(l)}$, respectively.
Then $X_{i}$ and $Y_{i}$ can be written in the following power series:\begin{subequations}\label{eq:X&Y}
\begin{eqnarray}
X_{i} & = & \sum_{p,q\geq0}(-1)^{p+q}\overleftarrow{X}_{p,q}^{(i)}\left(\frac{I_{\omega}\Delta r}{2}\right)^{p}\left(\frac{I_{\omega_{0}}\Delta r}{2}\right)^{q}\nonumber \\
 & = & \sum_{p\geq0}\overleftarrow{X}_{p}^{(i)}(\omega,\omega_{0})(\Delta r/2)^{p},\label{eq:Xseries}\\
Y_{i} & = & \sum_{p,q\geq0}(-1)^{p+q}\overleftarrow{Y}_{p,q}^{(i)}\left(\frac{I_{\omega}\Delta r}{2}\right)^{p}\left(\frac{I_{\omega_{0}}\Delta r}{2}\right)^{q}\nonumber \\
 & = & \sum_{p\geq0}\overleftarrow{Y}_{p}^{(i)}(\omega,\omega_{0})(\Delta r/2)^{p},\label{eq:Yseries}
\end{eqnarray}
where
\[
\overleftarrow{X}_{p}^{(i)}(\omega,\omega_{0})\equiv(-1)^{p}\sum_{q=0}^{p}\overleftarrow{X}_{p-q,q}^{(i)}I_{\omega}^{p-q}I_{\omega_{0}}^{q},
\]
\end{subequations}and $\overleftarrow{Y}_{p}^{(i)}(\omega,\omega_{0})$
is similarly defined. We can determine the coefficients $\overleftarrow{X}_{p,q}^{(i)}$
and $\overleftarrow{Y}_{p,q}^{(i)}$ from Eqs. (\ref{eq:etabynu})
and (\ref{eq:rho(n+1-omg)}), and typical ones are shown in Table.
\ref{tab:xpq&ypq}. From the table, we can find that $\overleftarrow{X}_{1}^{(i)}(\omega,\omega_{0})$
is independent of \textit{i}, and hereafter the superscript $(i)$
is omitted. Furthermore, $\hat{C}$ is regarded as the function of
$\rho$ because $Y_{i}$ converges to $\rho_{n}^{(l)}$ in the continuum
limit, $\Delta r\rightarrow0$, and, for example, we shall denote
$\left.(\partial\hat{C}/\partial\rho)\right|_{\rho=\rho_{n}^{(l)}}$
by $\hat{C}_{\rho}(\rho_{n}^{(l)})$. 
\begin{table}
\caption{Expansion coefficients $\protect\overleftarrow{X}_{p,q}^{(i)}$ and
$\protect\overleftarrow{Y}_{p,q}^{(i)}$ }
\label{tab:xpq&ypq}%
\begin{tabular}{ccccc}
\hline 
(\textit{p},\textit{q}) & $\overleftarrow{X}_{p,q}^{(1)}$ & $\overleftarrow{X}_{p,q}^{(2)}$ & $\overleftarrow{Y}_{p,q}^{(1)}$ & $\overleftarrow{Y}_{p,q}^{(2)}$\tabularnewline
\hline 
(0,0) & 0 & 0 & $\overleftarrow{\rho}_{n}^{(l)}$ & $\overleftarrow{\rho}_{n}^{(l)}$\tabularnewline
\hline 
(1,0) & $\frac{v}{D}\overleftarrow{\rho}_{n}^{(l)}$ & $\overleftarrow{X}_{1,0}^{(1)}$ & $\overleftarrow{X}_{1,0}^{(1)}$ & $\frac{v}{D}\overleftarrow{\rho}_{n}^{(l)}$\tabularnewline
 &  &  &  & $-2\overleftarrow{\left(\rho_{r}\right)}_{n}^{(l)}$\tabularnewline
\hline 
(0,1) & $-\frac{v}{D}\overleftarrow{\rho}_{n}^{(l)}$ & $\overleftarrow{X}_{0,1}^{(1)}$ & $-\overleftarrow{X}_{0,1}^{(1)}$ & $\overleftarrow{Y}_{0,1}^{(1)}$\tabularnewline
 & $+\overleftarrow{\left(\rho_{r}\right)}_{n}^{(l)}$ &  &  & \tabularnewline
\hline 
(2,0) & $\left(\frac{v}{D}\right)^{2}\overleftarrow{\rho}_{n}^{(l)}$ & $\left(\frac{v}{D}\right)^{2}\overleftarrow{\rho}_{n}^{(l)}$ & $\overleftarrow{X}_{2,0}^{(1)}$ & $\left(\frac{v}{D}\right)^{2}\overleftarrow{\rho}_{n}^{(l)}$\tabularnewline
 &  & $-\frac{2v}{D}\overleftarrow{\left(\rho_{r}\right)}_{n}^{(l)}$ &  & $-\frac{2v}{D}\overleftarrow{\left(\rho_{r}\right)}_{n}^{(l)}$\tabularnewline
 &  &  &  & $+2\overleftarrow{\left(\rho_{rr}\right)}_{n}^{(l)}$\tabularnewline
\hline 
(1,1) & 0 & $\frac{2v}{D}\overleftarrow{\left(\rho_{r}\right)}_{n}^{(l)}$ & 0 & $-\overleftarrow{X}_{1,1}^{(2)}$\tabularnewline
 &  & $-2\overleftarrow{\left(\rho_{rr}\right)}_{n}^{(l)}$ &  & \tabularnewline
\hline 
(0,2) & $-\left(\frac{v}{D}\right)^{2}\overleftarrow{\rho}_{n}^{(l)}$ & $\overleftarrow{X}_{0,2}^{(1)}$ & $-\overleftarrow{X}_{0,2}^{(1)}$ & $\overleftarrow{Y}_{0,2}^{(1)}$\tabularnewline
 & $+\frac{2v}{D}\overleftarrow{\left(\rho_{r}\right)}_{n}^{(l)}$ &  &  & \tabularnewline
 & $-\overleftarrow{\left(\rho_{rr}\right)}_{n}^{(l)}$ &  &  & \tabularnewline
\hline 
(3,0) & $\left(\frac{v}{D}\right)^{3}\overleftarrow{\rho}_{n}^{(l)}$ & $\left(\frac{v}{D}\right)^{3}\overleftarrow{\rho}_{n}^{(l)}$ & $\overleftarrow{X}_{3,0}^{(1)}$ & $\left(\frac{v}{D}\right)^{3}\overleftarrow{\rho}_{n}^{(l)}$\tabularnewline
 &  & $-2\left(\frac{v}{D}\right)^{2}\overleftarrow{\left(\rho_{r}\right)}_{n}^{(l)}$ &  & $-2\left(\frac{v}{D}\right)^{2}\overleftarrow{\left(\rho_{r}\right)}_{n}^{(l)}$\tabularnewline
 &  & $+\frac{2v}{D}\overleftarrow{\left(\rho_{rr}\right)}_{n}^{(l)}$ &  & $+\frac{2v}{D}\overleftarrow{\left(\rho_{rr}\right)}_{n}^{(l)}$\tabularnewline
 &  &  &  & $-\frac{4}{3}\overleftarrow{\left(\rho_{rrr}\right)}_{n}^{(l)}$\tabularnewline
\hline 
(2,1) & 0 & $-\frac{2v}{D}\overleftarrow{\left(\rho_{rr}\right)}_{n}^{(l)}$ & 0 & $-\overleftarrow{X}_{2,1}^{(2)}$\tabularnewline
 &  & $+2\overleftarrow{\left(\rho_{rrr}\right)}_{n}^{(l)}$ &  & \tabularnewline
\hline 
(1,2) & 0 & $2\left(\frac{v}{D}\right)^{2}\overleftarrow{\left(\rho_{r}\right)}_{n}^{(l)}$ & 0 & $-\overleftarrow{X}_{1,2}^{(2)}$\tabularnewline
 &  & $-\frac{4v}{D}\overleftarrow{\left(\rho_{rr}\right)}_{n}^{(l)}$ &  & \tabularnewline
 &  & $+2\overleftarrow{\left(\rho_{rrr}\right)}_{n}^{(l)}$ &  & \tabularnewline
\hline 
(0,3) & $-\left(\frac{v}{D}\right)^{3}\overleftarrow{\rho}_{n}^{(l)}$ & $\overleftarrow{X}_{0,3}^{(1)}$ & $-\overleftarrow{X}_{0,3}^{(1)}$ & $\overleftarrow{Y}_{0,3}^{(1)}$\tabularnewline
 & $+2\left(\frac{v}{D}\right)^{2}\overleftarrow{\left(\rho_{r}\right)}_{n}^{(l)}$ &  &  & \tabularnewline
 & $-\frac{2v}{D}\overleftarrow{\left(\rho_{rr}\right)}_{n}^{(l)}$ &  &  & \tabularnewline
 & $+\frac{2}{3}\overleftarrow{\left(\rho_{rrr}\right)}_{n}^{(l)}$ &  &  & \tabularnewline
\hline 
\end{tabular}
\end{table}

The following lemma is crucial for obtaining the leading-order terms
of $\left\langle A\right\rangle _{A_{n}}^{(l)}$.

\begin{lemma}\label{lem:F^(l)_n} $\mathsf{\mathbb{\mathfrak{\mathcal{F}}}}_{n}^{(l)}$,
Eq. (\ref{eq:defcalF^(l)_n}), has the following leading terms with
respect to $\Delta r$. This is the necessary and sufficient expansion
form in order to find the leading-order terms of $\left\langle A\right\rangle _{A_{n}}^{(l)}$.

\begin{multline}
\mathsf{\mathbb{\mathfrak{\mathcal{F}}}}_{n}^{(l)}(\epsilon,m,\omega,\underline{\omega}_{k})=(1+\epsilon(-1)^{m})\overleftarrow{X}_{1}(\omega,\omega_{0})^{m}\hat{C}(\overleftarrow{\rho}_{n}^{(l)})\left(\frac{\Delta r}{2}\right)^{m}\\
+\left[\overleftarrow{X}_{1}(\omega,\omega_{0})^{m}\hat{C}_{\rho}(\overleftarrow{\rho}_{n}^{(l)})\left(\overleftarrow{Y}_{1}^{(1)}(\omega,\omega_{0})+\epsilon(-1)^{m}\overleftarrow{Y}_{1}^{(2)}(\omega,\omega_{0})\right)\right.\\
\shoveleft+\left.m\overleftarrow{X}_{1}(\omega,\omega_{0})^{m-1}\hat{C}(\overleftarrow{\rho}_{n}^{(l)})\left(\overleftarrow{X}_{2}^{(1)}(\omega,\omega_{0})+\epsilon(-1)^{m}\overleftarrow{X}_{2}^{(2)}(\omega,\omega_{0})\right)\right]\left(\frac{\Delta r}{2}\right)^{m+1}\\
\shoveleft+\left\{ \overleftarrow{X}_{1}(\omega,\omega_{0})^{m}\left[\hat{C}_{\rho}(\overleftarrow{\rho}_{n}^{(l)})\left(\overleftarrow{Y}_{2}^{(1)}(\omega,\omega_{0})+\epsilon(-1)^{m}\overleftarrow{Y}_{2}^{(2)}(\omega,\omega_{0})\right)\right.\right.\\
\left.+\frac{1}{2}\hat{C}_{\rho\rho}(\overleftarrow{\rho}_{n}^{(l)})\left(\overleftarrow{Y}_{1}^{(1)}(\omega,\omega_{0})^{2}+\epsilon(-1)^{m}\overleftarrow{Y}_{1}^{(2)}(\omega,\omega_{0})^{2}\right)\right]\\
m\overleftarrow{X}_{1}(\omega,\omega_{0})^{m-1}\left[\hat{C}_{\rho}(\overleftarrow{\rho}_{n}^{(l)})\left(\overleftarrow{X}_{2}^{(1)}(\omega,\omega_{0})\overleftarrow{Y}_{1}^{(1)}(\omega,\omega_{0})\right.\right.\\
\left.+\epsilon(-1)^{m}\overleftarrow{X}_{2}^{(2)}(\omega,\omega_{0})\overleftarrow{Y}_{1}^{(2)}(\omega,\omega_{0})\right)\\
\left.\hat{C}(\overleftarrow{\rho}_{n}^{(l)})\left(\overleftarrow{X}_{3}^{(1)}(\omega,\omega_{0})+\epsilon(-1)^{m}\overleftarrow{X}_{3}^{(2)}(\omega,\omega_{0})\right)\right]\\
+\binom{m}{2}\overleftarrow{X}_{1}(\omega,\omega_{0})^{m-2}\hat{C}(\overleftarrow{\rho}_{n}^{(l)})\left(\overleftarrow{X}_{2}^{(1)}(\omega,\omega_{0})^{2}\right.\\
\left.\left.+\epsilon(-1)^{m}\overleftarrow{X}_{2}^{(2)}(\omega,\omega_{0})^{2}\right)\right\} \left(\frac{\Delta r}{2}\right)^{m+2}+O(\Delta r^{m+3}).\label{eq:calF^(l)_n}
\end{multline}

\end{lemma}

\begin{proof}First of all, we should notice that higher-order terms
in Eq. (\ref{eq:expansion_f}), i.e. $C^{(m+2)}(Y)X^{m+2}$, $C^{(m+4)}(Y)X^{m+4}$,$\cdots$,
does not affect the leading-order terms of $\left\langle A\right\rangle _{A_{n}}^{(l)}$.
From the term $C^{(m+2)}(Y)X^{m+2}$, for example, (\textit{m}+2)th
order terms appears. If the leading \textit{m}th order terms, originated
from the term $\hat{C}(Y)X^{m}$, vanish, the (\textit{m}+2)th order
terms also do because 
\[
1+\epsilon(-1)^{m}=1+\epsilon(-1)^{m+2}(=1+\epsilon(-1)^{m+4}=\cdots).
\]
Similarly, when the leading-order terms of $\hat{C}(Y)X^{m}$ increase
its order through the operations (\ref{eq:S_nu(a)}) and (\ref{eq:S_etanu}),
those of $C^{(m+2)}(Y)X^{m+2}$ also do because the same parameters
of $\epsilon$ and $\delta$ are shared. Therefore, the difference
of order exponents between the two leading-order terms originated
from $\hat{C}(Y)X^{m}$ and $C^{(m+2)}(Y)X^{m+2}$ is always kept
constant of two. Similar discussion can be applied to the leading
terms originated from the order-disjoint two terms $C^{(m+2a)}(Y)X^{m+2a}$
and $C^{(m+2a+2)}(Y)X^{m+2a+2}$. Consequently, the order exponents
of these leading terms are equally spaced, with an interval of two.
That is why the function \textit{f} can be replaced by $\hat{C}(Y)X^{m}$
in order to obtain the specific leading terms of $\left\langle A\right\rangle _{A_{n}}^{(l)}$.
Replacing \textit{f} with $\hat{C}(Y)X^{m}$ and substituting Eq.
(\ref{eq:X&Y}) into $\mathcal{F}_{n}^{(l)}$, Eq. (\ref{eq:defcalF^(l)_n}),
yields Eq. (\ref{eq:calF^(l)_n}). 

As described above we should treat (\textit{m}+3)th order terms for
the case (II-2), which is the same order as the higher-order term
$O(\Delta r^{m+3})$ in Eq. (\ref{eq:calF^(l)_n}). In this case,
however, the operations (\ref{eq:S_nu(a)}) and (\ref{eq:S_etanu})
on the higher-order term make the same order terms vanish because
of the coefficient $(\delta+1)$. That is why the expansion (\ref{eq:calF^(l)_n})
is necessary and sufficient to find the leading-order terms of \textit{K}
or $\left\langle A\right\rangle _{A_{n}}^{(l)}$. \end{proof}

\begin{corollary}The LOE (Theorem \ref{theorem:LOE} or Eq. (\ref{eq:LOE}))
immediately follows from the Lemma \ref{lem:F^(l)_n} (or Eq. (\ref{eq:calF^(l)_n}))
and properties (\ref{eq:S_nu(a)}), (\ref{eq:S_etanu}) and (\ref{eq:sum_nugA}).\end{corollary}

Now we are in a position to have the leading-order terms of the averaged
quantity $\left\langle A\right\rangle _{A_{n}}^{(l)}$. In what follows,
some variables are defined as follows

\[
\, X_{1,0}^{(1)}=X_{1,0}^{(2)}\equiv\frac{v}{D}\rho_{n}^{(l)}\equiv X_{1,0},
\]
\[
X_{0,1}^{(1)}=X_{0,1}^{(2)}\equiv-\frac{v}{D}\rho_{n}^{(l)}+(\rho_{r})_{n}^{(l)}\equiv X_{0,1},
\]
and this is the replacement of $\overleftarrow{\rho}_{n}^{(l)}$ and
$\overleftarrow{\left(\rho_{r}\right)}_{n}^{(l)}$ in $\overleftarrow{X}_{0,1}^{(\cdot)}$
or $\overleftarrow{X}_{1,0}^{(\cdot)}$ (cf. Table. \ref{tab:xpq&ypq})
with $\rho_{n}^{(l)}$ and $(\rho_{r})_{n}^{(l)}$, respectively.
Other quantities $X_{p,q}^{(i)}$ and $Y_{p,q}^{(i)}$ are similarly
defined. These operations are originated from the property (\ref{eq:sum_nugA}).
$\rho_{n}^{(l)}$ and $(\rho_{r})_{n}^{(l)}$ are identified with
$\rho(r,t)$ and $(\partial\rho/\partial r)(r,t)$, respectively,
when we regard these quantities as those defined at a position $r=n\Delta r$
and time $t=l\tau$. 

In addition, we assume that $\sum_{p=a}^{b}A_{p}=0$ when \textit{$b<a$}.
The leading terms often include the form of $\left[(A)^{(1)}+\epsilon(-1)^{m}(A)^{(2)}\right]$
(see Case (II-1) below). In what follows, therefore, only the first
part $(A)^{(1)}$ is described and the second counterpart is often
abridged to $(\,^{(2)})$.

\begin{theorem}[\bf{Leading-order form}]\begin{subequations}\label{eq:Leadingform}Consider
the averaged value $\left\langle A\right\rangle _{A_{n}}^{(l)}$ of
a \textit{$T^{1}$}-type local quantity $A^{(l)}(f,m,\hat{C})$ for
a UOFP type multibaker chain system partitioned by a Markov partition
$(n,\underline{\omega}_{k})$. Its leading-order terms have the following
form:

\subsubsection*{Case (I-1): (m-2)th order term}

For the case of \textit{m}=0 
\[
\frac{K}{\tau\Delta r}\cong\frac{4\hat{C}(\rho_{n}^{(l)})}{\tau},
\]
and for the case $m=2m'\,(m'>0)$ 
\begin{equation}
\frac{K}{\tau\Delta r}\cong\frac{D\hat{C}(\rho_{n}^{(l)})}{2^{m-3}}\left(X_{0,1}^{m}+X_{1,0}^{m}+\beta\sum_{p=1}^{m'-1}\binom{m}{2p}X_{1,0}^{2p}X_{0,1}^{2(m'-p)}\right)\Delta r^{m-2}.\label{eq:Leadform_I1}
\end{equation}

\subsubsection*{Case (II-1): (m-1)th order term}

For the case of $m=2m'+1\,(m'\geq0)$ we obtain

\begin{multline}
\frac{K}{\tau\Delta r}\cong-(1+\epsilon(-1)^{m})\frac{v\hat{C}(\rho_{n}^{(l)})}{2^{m}}\left[(\delta+1)\sum_{p=0}^{m'}\binom{m}{2p}X_{1,0}^{2p}X_{0,1}^{2(m'-p)+1}S_{\nu}(2p)\right.\\
\left.+(\delta-1)\sum_{p=0}^{m'}\binom{m}{2p+1}X_{1,0}^{2p+1}X_{0,1}^{2(m'-p)}S_{\nu}(2(m'-p))\right]\Delta r^{m-1}\\
\shoveleft+(\delta+1)\frac{D}{2^{m}}\hat{C}_{\rho}(\rho_{n}^{(l)})\left[\sum_{p=0}^{m'}\binom{m}{2p}X_{1,0}^{2p}X_{0,1}^{2(m'-p)+1}(Y_{0,1}^{(1)}-\epsilon Y_{0,1}^{(2)})S_{\nu}(2p)\right.\\
\left.+\sum_{p=0}^{m'}\binom{m}{2p+1}X_{1,0}^{2p+1}X_{0,1}^{2(m'-p)}(Y_{1,0}^{(1)}-\epsilon Y_{1,0}^{(2)})S_{\nu}(2(m'-p))\right]\Delta r^{m-1}\\
\shoveleft+(\delta+1)\frac{mD}{2^{m}}\hat{C}(\rho_{n}^{(l)})\left\{ \sum_{p=0}^{m'}\binom{2m'}{2p}X_{1,0}^{2p}X_{0,1}^{2(m'-p)}\left[X_{2,0}^{(1)}S_{\nu}(2(m'-p))+X_{0,2}^{(1)}S_{\nu}(2p)\right.\right.\\
\left.-\epsilon\left(X_{2,0}^{(2)}S_{\nu}(2(m'-p))+X_{0,2}^{(2)}S_{\nu}(2p)\right)\right]\\
\left.+\beta\sum_{p=0}^{m'-1}\binom{2m'}{2p+1}X_{1,0}^{2p+1}X_{0,1}^{2(m'-p)-1}(X_{1,1}^{(1)}-\epsilon X_{1,1}^{(2)})\right\} \Delta r^{m-1}.\label{eq:Leadform_II-1}
\end{multline}

\subsubsection*{Case (I-2): mth order term}

For $m=2m'\,(m'\geq0)$ the leading terms are as follows:
\begin{multline}
\frac{K}{\tau\Delta r}\cong\beta\frac{(1+\epsilon)(\delta-1)}{2^{m+1}}\frac{v^{2}\hat{C}(\rho_{n}^{(l)})}{D}\sum_{p=0}^{m'-1}\binom{m}{2p+1}X_{1,0}^{2p+1}X_{0,1}^{2(m'-p)-1}\Delta r^{m}\\
\shoveleft-\frac{v}{2^{m+1}}\hat{C}_{\rho}(\rho_{n}^{(l)})\left\{ \sum_{p=0}^{m'}\binom{m}{2p}X_{1,0}^{2p}X_{0,1}^{2(m'-p)}\left[(\delta+1)Y_{0,1}^{(1)}S_{\nu}(2p)\right.\right.\\
\left.+(\delta-1)Y_{1,0}^{(1)}S_{\nu}(2(m'-p))+\epsilon(\,^{(2)})\right]\\
+\beta\sum_{p=0}^{m'-1}\binom{m}{2p+1}X_{1,0}^{2p+1}X_{0,1}^{2(m'-p)-1}\\
\left.\text{\ensuremath{\times}}\left[(\delta-1)Y_{0,1}^{(1)}+(\delta+1)Y_{1,0}^{(1)}+\epsilon(\,^{(2)})\right]\right\} \Delta r^{m}\\
\displaybreak[3]\\
\shoveleft-\frac{m}{2^{m+1}}v\hat{C}(\rho_{n}^{(l)})\left\{ \sum_{p=0}^{m'-1}\binom{m-1}{2p}X_{1,0}^{2p}X_{0,1}^{2(m'-p)-1}\left[(\delta+1)X_{0,2}^{(1)}S_{\nu}(2p)\right.\right.\\
\left.+\beta(\delta-1)X_{1,1}^{(1)}+\beta(\delta+1)X_{2,0}^{(1)}+\epsilon(\,^{(2)})\right]\\
+\sum_{p=0}^{m'-1}\binom{m-1}{2p+1}X_{1,0}^{2p+1}X_{0,1}^{2(m'-p-1)}\left[\beta(\delta-1)X_{0,2}^{(1)}\right.\\
+\beta(\delta+1)X_{1,1}^{(1)}+(\delta-1)X_{2,0}^{(1)}S_{\nu}(2(m'-p-1))\\
\left.\left.+\epsilon(\,^{(2)})\right]\right\} \Delta r^{m}\\
\shoveleft+(\delta+1)\frac{D}{2^{m+1}}\left\{ \sum_{p=0}^{m'}\binom{m}{2p}X_{1,0}^{2p}X_{0,1}^{2(m'-p)}\left[\hat{C}_{\rho}(\rho_{n}^{(l)})\left(Y_{0,2}^{(1)}S_{\nu}(2p)\right.\right.\right.\\
\left.+Y_{2,0}^{(1)}S_{\nu}(2(m'-p))+\epsilon(\,^{(2)})\right)\\
\left.+\frac{1}{2}\hat{C}_{\rho\rho}(\rho_{n}^{(l)})\left(Y_{0,1}^{(1)^{2}}S_{\nu}(2p)+Y_{1,0}^{(1)^{2}}S_{\nu}(2(m'-p))+\epsilon(\,^{(2)})\right)\right]\\
+\beta\sum_{p=0}^{m'-1}\binom{m}{2p+1}X_{1,0}^{2p+1}X_{0,1}^{2(m'-p)-1}\left[\hat{C}_{\rho}(\rho_{n}^{(l)})\left(Y_{1,1}^{(1)}+\epsilon Y_{1,1}^{(2)}\right)\right.\\
\left.\left.+\hat{C}_{\rho\rho}(\rho_{n}^{(l)})\left(Y_{0,1}^{(1)}Y_{1,0}^{(1)}+\epsilon(\,^{(2)})\right)\right]\right\} \Delta r^{m}\\
\shoveleft+(\delta+1)\frac{mD}{2^{m+1}}\left\{ \sum_{p=0}^{m'-1}\binom{m-1}{2p}X_{1,0}^{2p}X_{0,1}^{2(m'-p)-1}\left[\hat{C}_{\rho}(\rho_{n}^{(l)})\left(X_{0,2}^{(1)}Y_{0,1}^{(1)}S_{\nu}(2p)\right.\right.\right.\\
\left.+\beta X_{1,1}^{(1)}Y_{1,0}^{(1)}+\beta X_{2,0}^{(1)}Y_{0,1}^{(1)}+\epsilon(\,^{(2)})\right)+\hat{C}(\rho_{n}^{(l)})\left(X_{0,3}^{(1)}S_{\nu}(2p)\right.\\
\left.\left.+\beta X_{2,1}^{(1)}+\epsilon(\,^{(2)})\right)\right]+\sum_{p=0}^{m'-1}\binom{m-1}{2p+1}X_{1,0}^{2p+1}X_{0,1}^{2(m'-p-1)}\\
\times\left[\hat{C}_{\rho}(\rho_{n}^{(l)})\left(\beta X_{0,2}^{(1)}Y_{1,0}^{(1)}+\beta X_{1,1}^{(1)}Y_{0,1}^{(1)}+X_{2,0}^{(1)}Y_{1,0}^{(1)}S_{\nu}(2(m'-p-1))\right.\right.\\
\left.\left.\left.+\epsilon(\,^{(2)})\right)+\hat{C}(\rho_{n}^{(l)})\left(\beta X_{1,2}^{(1)}+X_{3,0}^{(1)}S_{\nu}(2(m'-p-1))+\epsilon(\,^{(2)})\right)\right]\right\} \Delta r^{m}\\
\shoveleft+(\delta+1)\frac{m(m-1)}{2^{m+2}}D\hat{C}(\rho_{n}^{(l)})\left\{ \sum_{p=0}^{m'-1}\binom{m-2}{2p}X_{1,0}^{2p}X_{0,1}^{2(m'-p-1)}\left[X_{0,2}^{(1)^{2}}S_{\nu}(2p)\right.\right.\\
\left.+2\beta X_{0,2}^{(1)}X_{2,0}^{(1)}+\beta X_{1,1}^{(1)^{2}}+X_{2,0}^{(1)^{2}}S_{\nu}(2(m'-p-1))+\epsilon(\,^{(2)})\right]\\
+\beta\sum_{p=0}^{m'-2}\binom{m-2}{2p+1}X_{1,0}^{2p+1}X_{0,1}^{2(m'-p-1)-1}\left[2X_{0,2}^{(1)}X_{1,1}^{(1)}+2X_{1,1}^{(1)}X_{2,0}^{(1)}\right.\\
\left.\left.+\epsilon(\,^{(2)})\right]\right\} \Delta r^{m}.\label{eq:Leadform_I-2}
\end{multline}

\subsubsection*{Case (II-2):(m+1)th order term}

For the case of $m=2m'+1\,(m'>0)$ 

\begin{multline}
\frac{K}{\tau\Delta r}\cong-\frac{\beta}{2^{m+1}}\frac{v^{2}\hat{C_{\rho}}(\rho_{n}^{(l)})}{D}\left[\sum_{p=0}^{m'}\binom{m}{2p}X_{1,0}^{2p}X_{0,1}^{2(m'-p)+1}(Y_{1,0}^{(1)}-Y_{1,0}^{(2)})\right.\\
\left.+\sum_{p=0}^{m'}\binom{m}{2p+1}X_{1,0}^{2p+1}X_{0,1}^{2(m'-p)}(Y_{0,1}^{(1)}-Y_{0,1}^{(2)})\right]\Delta r^{m+1}\\
\shoveleft-\frac{m\beta}{2^{m+1}}\frac{v^{2}\hat{C}(\rho_{n}^{(l)})}{D}\left[\sum_{p=0}^{m'}\binom{m-1}{2p}X_{1,0}^{2p}X_{0,1}^{2(m'-p)}(X_{1,1}^{(1)}-X_{1,1}^{(2)})\right.\\
\left.+\sum_{p=0}^{m'-1}\binom{m-1}{2p+1}X_{1,0}^{2p+1}X_{0,1}^{2(m'-p)-1}\left(X_{0,2}^{(1)}+X_{2,0}^{(1)}-(\,^{(2)})\right)\right]\Delta r^{m+1}\\
\shoveleft+\frac{v}{2^{m+1}}\left\{ \beta\sum_{p=0}^{m'}\binom{m}{2p}X_{1,0}^{2p}X_{0,1}^{2(m'-p)+1}\left[\hat{C}_{\rho}(\rho_{n}^{(l)})(Y_{1,1}^{(1)}-Y_{1,1}^{(2)})\right.\right.\\
\left.+\hat{C}_{\rho\rho}(\rho_{n}^{(l)})(Y_{0,1}^{(1)}Y_{1,0}^{(1)}-(\,^{(2)}))\right]\\
+\sum_{p=0}^{m'}\binom{m}{2p+1}X_{1,0}^{2p+1}X_{0,1}^{2(m'-p)}\left[\hat{C}_{\rho}(\rho_{n}^{(l)})\left(\beta Y_{0,2}^{(1)}+Y_{2,0}^{(1)}S_{\nu}(2(m'-p))\right.\right.\\
\left.\left.\left.-(\,^{(2)})\right)+\frac{1}{2}\hat{C}_{\rho\rho}(\rho_{n}^{(l)})\left(\beta Y_{0,1}^{(1)^{2}}+Y_{1,0}^{(1)^{2}}S_{\nu}(2(m'-p))-(\,^{(2)})\right)\right]\right\} \Delta r^{m+1}\\
\shoveleft+\frac{mv}{2^{m+1}}\left\{ \sum_{p=0}^{m'}\binom{m-1}{2p}X_{1,0}^{2p}X_{0,1}^{2(m'-p)}\left[\hat{C}_{\rho}(\rho_{n}^{(l)})\left(\beta X_{0,2}^{(1)}Y_{1,0}^{(1)}\right.\right.\right.\\
\left.+\beta X_{1,1}^{(1)}Y_{0,1}^{(1)}+X_{2,0}^{(1)}Y_{1,0}^{(1)}S_{\nu}(2(m'-p))-(\,^{(2)})\right)\\
\left.+\hat{C}(\rho_{n}^{(l)})\left(\beta X_{1,2}^{(1)}+X_{3,0}^{(1)}S_{\nu}(2(m'-p))-(\,^{(2)})\right)\right]\\
+\beta\sum_{p=0}^{m'-1}\binom{m-1}{2p+1}X_{1,0}^{2p+1}X_{0,1}^{2(m'-p)-1}\left[\hat{C}_{\rho}(\rho_{n}^{(l)})\left(X_{0,2}^{(1)}Y_{0,1}^{(1)}+X_{1,1}^{(1)}Y_{1,0}^{(1)}\right.\right.\\
\left.\left.\left.+X_{2,0}^{(1)}Y_{0,1}^{(1)}-(\,^{(2)})\right)+\hat{C}(\rho_{n}^{(l)})\left(X_{2,1}^{(1)}+X_{0,3}^{(1)}-(\,^{(2)})\right)\right]\right\} \Delta r^{m+1}\\
\displaybreak[3]\\
\shoveleft+\frac{m(m-1)}{2^{m+2}}v\hat{C}(\rho_{n}^{(l)})\left[2\beta\sum_{p=0}^{m'-1}\binom{m-2}{2p}X_{1,0}^{2p}X_{0,1}^{2(m'-p)-1}\left(X_{0,2}^{(1)}X_{1,1}^{(1)}\right.\right.\\
\left.+X_{1,1}^{(1)}X_{2,0}^{(1)}-(\,^{(2)})\right)+\sum_{p=0}^{m'-1}\binom{m-2}{2p+1}X_{1,0}^{2p+1}X_{0,1}^{2(m'-p-1)}\\
\times\left(\beta X_{0,2}^{(1)^{2}}+2\beta X_{0,2}^{(1)}X_{2,0}^{(1)}+\beta X_{1,1}^{(1)^{2}}+X_{2,0}^{(1)^{2}}S_{\nu}(2(m'-p-1))\right.\\
\left.\left.-(\,^{(2)})\right)\right]\Delta r^{m+1}.\label{eq:Leadform_II-2}
\end{multline}

\end{subequations}\end{theorem}

\begin{proof}follows immediately from Lemma \ref{lem:F^(l)_n}, i.e.
the operations (\ref{eq:S_nu(a)}), (\ref{eq:S_etanu}) and (\ref{eq:sum_nugA})
on the operand (summand) $\mathcal{F}_{n}^{(l)}$, Eq. (\ref{eq:calF^(l)_n}).\end{proof}

\begin{remark}We can find that the leading terms (\ref{eq:Leadingform})
involves $\beta$ and, therefore, its quantity depends on a coarse-graining,
partitioning parameter: it depends on how the time step $\tau$ approaches
to zero in the continuum limit, $\Delta r\rightarrow0$. It is easy
for us to convince or show that the second leading-order terms are
also dependent on $\beta$. We can numerically confirm the dependency
when $\Delta r$ approaches to zero. The parameter $\beta$, however,
is nothing to do with the leading order, estimated from the LOE. Thus
the order itself is independent of $\beta$. Moreover, we can find
that the leading terms are independent of $\beta$ if and only if
their order is $O(1)$. Such an ``observable'' order is achieved
for an appropriate characteristic exponent \textit{m} except the case
(II-2) because \textit{m} cannot be negative integer. It follows that
its limit value in the continuum limit, i.e. macroscopic quantity,
is also independent of the parameter $\beta$. \end{remark}

\begin{remark}We also notice that the leading terms are also independent
of the other partitioning parameter, digit number \textit{k} for $k\geq1$.
As the results, the LOE and the macroscopic quantities are independent
of partitioning parameters. This is essential for the coarse-grained
dynamics to be consistent with irreversible thermodynamics \cite{vollmers,gilbert1}.
The least number of $k=1$ corresponds to the approach of level-0
partitioning, utilized by Vollmer, T\'{e}l, M\'{a}ty\'{a}s, et al.
\cite{vollmer1,vollmer5,vollmer3,cross1,cross2,cross3,vollmer2,vollmers}.
The property comes from the fact that the leading terms, as well as
the LOE, are determined by the properties related to only two trinary
bits $\omega$ and $\omega_{0}$ (Lemma \ref{lem:order_K}). For the
case of \textit{k}=0, one of the trinary bit $\omega_{0}$ vanishes
and, therefore, another LOE or leading form appears. However, Vollmer,
T\'{e}l, and Breymann \cite{vollmer1} showed that such a ``projected''
dynamics is not physically relevant in the sense that it cannot explain
the effects of thermostat.\end{remark}

\begin{example}Ishida \cite{ishida2012} showed that the balance
equation for local entropy density $s_{i}^{(l)}(\equiv-\rho_{i}^{(l)}\ln(\rho_{i}^{(l)}/\rho_{r}))$
can be expressed as follows:
\begin{eqnarray}
\Delta V_{i}\Delta s_{i}^{(l)} & \equiv & \Delta V_{i}(s_{i}^{(l+1)}-s_{i}^{(l)})\nonumber \\
 & = & \hat{S}_{c,i}^{(l)}+\hat{S}_{d,i}^{(l)}+\hat{S}_{p,i}^{(l)}+\hat{S}_{th,i}^{(l)}+r_{T,i}^{(l)},\label{eq:localentbalance}
\end{eqnarray}
where

\begin{eqnarray*}
\hat{S}_{c,i}^{(l)} & \equiv & \sum_{j}\left[\frac{J_{j\leftarrow i}^{[a](l)}}{\Delta_{i,j}}\left(1+\frac{\ln(\rho_{j}^{(l)}/\rho_{r})+\ln(\rho_{i}^{(l)}/\rho_{r})}{2}\right)+\frac{\tilde{D}_{i,j}^{[a]}}{\Delta_{i,j}^{2}}(\rho_{j}^{(l)}-\rho_{i}^{(l)})\right],\\
\hat{S}_{d,i}^{(l)} & \equiv & \sum_{j}\left[\frac{J_{j\leftarrow i}^{[s](l)}}{\Delta_{i,j}}\left(1+\frac{\ln(\rho_{j}^{(l)}/\rho_{r})+\ln(\rho_{i}^{(l)}/\rho_{r})}{2}\right)-\frac{\tilde{U}_{j\leftarrow i}^{[s]}}{\Delta_{i,j}}\frac{\rho_{j}^{(l)}+\rho_{i}^{(l)}}{2}\right],\\
\hat{S}_{p,i}^{(l)} & \equiv & \sum_{j}\left(-\frac{J_{j\leftarrow i}^{[\hat{a}](l)}}{\Delta_{i,j}}\frac{\ln\rho_{j}^{(l)}-\ln\rho_{i}^{(l)}}{2}\right),\\
\hat{S}_{th,i}^{(l)} & \equiv & \sum_{j}\left[-\frac{J_{j\leftarrow i}^{[\hat{s}](l)}}{\Delta_{i,j}}\frac{\ln\rho_{j}^{(l)}-\ln\rho_{i}^{(l)}}{2}\right.\\
 &  & \left.+\frac{1}{\Delta_{i,j}}\left(\tilde{U}_{j\leftarrow i}^{[s]}\frac{\rho_{j}^{(l)}-\rho_{i}^{(l)}}{2}-\frac{\tilde{D}_{i,j}^{[a]}}{\Delta_{i,j}}(\rho_{j}^{(l)}+\rho_{i}^{(l)})\right)\right],\\
r_{T,i}^{(l)} & \equiv & -2\Delta V_{i}\rho_{i}^{(l+1)}\left(\frac{\ln\rho_{i}^{(l+1)}-\ln\rho_{i}^{(l)}}{2}-\frac{\Delta\rho_{i}^{(l)}}{2\rho_{i}^{(l+1)}}\right).
\end{eqnarray*}
By use of the leading terms (\ref{eq:Leadingform}), and LOE (\ref{eq:LOE}),
we can now derive the averaged form of each component of the above
local entropy change in the bulk system as follows:
\begin{eqnarray*}
\left\langle \hat{S}_{c}\right\rangle _{A_{n}}^{(l)} & = & -\frac{\partial}{\partial r}(sv-\rho v)+O(\Delta r^{2}),\\
\left\langle \hat{S}_{d}\right\rangle _{A_{n}}^{(l)} & = & D\frac{\partial^{2}s}{\partial r^{2}}+O(\Delta r^{2}),\\
\left\langle \hat{S}_{p}\right\rangle _{A_{n}}^{(l)} & = & \frac{j^{2}}{\rho D}+O(\Delta r^{2}),\\
\left\langle \hat{S}_{th}\right\rangle _{A_{n}}^{(l)} & = & -\frac{jv}{D}+O(\Delta r^{2}),
\end{eqnarray*}
where $s\equiv-\rho\ln(\rho/\rho_{r}),\, j\equiv\rho v-D\partial\rho/\partial r$.
Herein $\left\langle A\right\rangle _{A_{n}}^{(l)}$ denotes the averaged
quantity on $A_{n}$ at the \textit{l}th time step, interpreted as
the quantity at $r=n\Delta r$ and $t=l\tau$ on its right side. Considering
that the averaging of $\Delta V_{i}\Delta s_{i}^{(l)}$ is the time
derivative of the averaged entropy density on $A_{n}$ and that $\left\langle r_{T}\right\rangle _{A_{n}}^{(l)}$
is $O(\tau)$ \cite{ishida2012}, the averaging of Eq. (\ref{eq:localentbalance})
recovers the macroscopic entropy balance equation of irreversible
thermodynamics in the continuum limit, $\Delta r\rightarrow0$. It
has been numerically confirmed by Ishida \cite{ishida2012}.\end{example}

\section{Concluding remarks}

For $T^{1}$-type local quantities on the UOFP multibaker chain system,
a law of order estimation (LOE) to estimate the order exponent of
the averaged local quantities, originally conjectured by Ishida\cite{ishida2012},
becomes a theorem. Furthermore, the form of the leading-order terms
for the quantities is derived, and we can confirm the order and the
finite limit value in the continuum limit are independent of partitioning
parameters, such as $\beta$ and trit number \textit{k}. The results
fully explain the numerical results in the bulk system, obtained by
Ishida \cite{ishida1,ishida2012}, and they are consistent with the
irreversible thermodynamics. 
\begin{acknowledgments}
I am grateful to Joel L. Lebowitz and Sheldon Goldstein for enlightening
discussions. This work was partially supported by the Japan Society
for the Promotion of Science, Institutional Program 'Development of
International Network for Training of Young Researchers Exploring
Multidisciplinary Fields', No. J091113016.
\end{acknowledgments}
\bibliographystyle{aipnum4-1}

\begin{thebibliography}{20}%
\makeatletter
\providecommand \@ifxundefined [1]{%
 \@ifx{#1\undefined}
}%
\providecommand \@ifnum [1]{%
 \ifnum #1\expandafter \@firstoftwo
 \else \expandafter \@secondoftwo
 \fi
}%
\providecommand \@ifx [1]{%
 \ifx #1\expandafter \@firstoftwo
 \else \expandafter \@secondoftwo
 \fi
}%
\providecommand \natexlab [1]{#1}%
\providecommand \enquote  [1]{``#1''}%
\providecommand \bibnamefont  [1]{#1}%
\providecommand \bibfnamefont [1]{#1}%
\providecommand \citenamefont [1]{#1}%
\providecommand \href@noop [0]{\@secondoftwo}%
\providecommand \href [0]{\begingroup \@sanitize@url \@href}%
\providecommand \@href[1]{\@@startlink{#1}\@@href}%
\providecommand \@@href[1]{\endgroup#1\@@endlink}%
\providecommand \@sanitize@url [0]{\catcode `\\12\catcode `\$12\catcode
  `\&12\catcode `\#12\catcode `\^12\catcode `\_12\catcode `\%12\relax}%
\providecommand \@@startlink[1]{}%
\providecommand \@@endlink[0]{}%
\providecommand \url  [0]{\begingroup\@sanitize@url \@url }%
\providecommand \@url [1]{\endgroup\@href {#1}{\urlprefix }}%
\providecommand \urlprefix  [0]{URL }%
\providecommand \Eprint [0]{\href }%
\providecommand \doibase [0]{http://dx.doi.org/}%
\providecommand \selectlanguage [0]{\@gobble}%
\providecommand \bibinfo  [0]{\@secondoftwo}%
\providecommand \bibfield  [0]{\@secondoftwo}%
\providecommand \translation [1]{[#1]}%
\providecommand \BibitemOpen [0]{}%
\providecommand \bibitemStop [0]{}%
\providecommand \bibitemNoStop [0]{.\EOS\space}%
\providecommand \EOS [0]{\spacefactor3000\relax}%
\providecommand \BibitemShut  [1]{\csname bibitem#1\endcsname}%
\let\auto@bib@innerbib\@empty
%</preamble>
\bibitem [{\citenamefont {Gilbert}\ and\ \citenamefont
  {Dorfman}(1999)}]{gilbert2}%
  \BibitemOpen
  \bibfield  {author} {\bibinfo {author} {\bibfnamefont {T.}~\bibnamefont
  {Gilbert}}\ and\ \bibinfo {author} {\bibfnamefont {J.}~\bibnamefont
  {Dorfman}},\ }\href@noop {} {\bibfield  {journal} {\bibinfo  {journal} {J.
  Stat. Phys.}\ }\textbf {\bibinfo {volume} {96}},\ \bibinfo {pages} {225}
  (\bibinfo {year} {1999})}\BibitemShut {NoStop}%
\bibitem [{\citenamefont {Vollmer}, \citenamefont {T{\'e}l},\ and\
  \citenamefont {Breymann}(1998)}]{vollmer1}%
  \BibitemOpen
  \bibfield  {author} {\bibinfo {author} {\bibfnamefont {J.}~\bibnamefont
  {Vollmer}}, \bibinfo {author} {\bibfnamefont {T.}~\bibnamefont {T{\'e}l}}, \
  and\ \bibinfo {author} {\bibfnamefont {W.}~\bibnamefont {Breymann}},\
  }\href@noop {} {\bibfield  {journal} {\bibinfo  {journal} {Phys. Rev. E}\
  }\textbf {\bibinfo {volume} {58}},\ \bibinfo {pages} {1672} (\bibinfo {year}
  {1998})}\BibitemShut {NoStop}%
\bibitem [{\citenamefont {Breymann}, \citenamefont {T{\'e}l},\ and\
  \citenamefont {Vollmer}(1998)}]{vollmer5}%
  \BibitemOpen
  \bibfield  {author} {\bibinfo {author} {\bibfnamefont {W.}~\bibnamefont
  {Breymann}}, \bibinfo {author} {\bibfnamefont {T.}~\bibnamefont {T{\'e}l}}, \
  and\ \bibinfo {author} {\bibfnamefont {J.}~\bibnamefont {Vollmer}},\
  }\href@noop {} {\bibfield  {journal} {\bibinfo  {journal} {Chaos}\ }\textbf
  {\bibinfo {volume} {8}},\ \bibinfo {pages} {396} (\bibinfo {year}
  {1998})}\BibitemShut {NoStop}%
\bibitem [{\citenamefont {Vollmer}, \citenamefont {T{\'e}l},\ and\
  \citenamefont {Breymann}(1997)}]{vollmer3}%
  \BibitemOpen
  \bibfield  {author} {\bibinfo {author} {\bibfnamefont {J.}~\bibnamefont
  {Vollmer}}, \bibinfo {author} {\bibfnamefont {T.}~\bibnamefont {T{\'e}l}}, \
  and\ \bibinfo {author} {\bibfnamefont {W.}~\bibnamefont {Breymann}},\
  }\href@noop {} {\bibfield  {journal} {\bibinfo  {journal} {Phys. Rev. Lett.}\
  }\textbf {\bibinfo {volume} {79}},\ \bibinfo {pages} {2759} (\bibinfo {year}
  {1997})}\BibitemShut {NoStop}%
\bibitem [{\citenamefont {T{\'e}l}, \citenamefont {Vollmer},\ and\
  \citenamefont {Breymann}(1996)}]{vollmer2}%
  \BibitemOpen
  \bibfield  {author} {\bibinfo {author} {\bibfnamefont {T.}~\bibnamefont
  {T{\'e}l}}, \bibinfo {author} {\bibfnamefont {J.}~\bibnamefont {Vollmer}}, \
  and\ \bibinfo {author} {\bibfnamefont {W.}~\bibnamefont {Breymann}},\
  }\href@noop {} {\bibfield  {journal} {\bibinfo  {journal} {Europhys. Lett.}\
  }\textbf {\bibinfo {volume} {35}},\ \bibinfo {pages} {659} (\bibinfo {year}
  {1996})}\BibitemShut {NoStop}%
\bibitem [{\citenamefont {Ishida}(2009)}]{ishida1}%
  \BibitemOpen
  \bibfield  {author} {\bibinfo {author} {\bibfnamefont {H.}~\bibnamefont
  {Ishida}},\ }\href@noop {} {\bibfield  {journal} {\bibinfo  {journal}
  {Physica A}\ }\textbf {\bibinfo {volume} {388}},\ \bibinfo {pages} {332}
  (\bibinfo {year} {2009})}\BibitemShut {NoStop}%
\bibitem [{\citenamefont {Tasaki}\ and\ \citenamefont
  {Gaspard}(1995)}]{tasaki1}%
  \BibitemOpen
  \bibfield  {author} {\bibinfo {author} {\bibfnamefont {S.}~\bibnamefont
  {Tasaki}}\ and\ \bibinfo {author} {\bibfnamefont {P.}~\bibnamefont
  {Gaspard}},\ }\href@noop {} {\bibfield  {journal} {\bibinfo  {journal} {J.
  Stat. Phys.}\ }\textbf {\bibinfo {volume} {81}},\ \bibinfo {pages} {935}
  (\bibinfo {year} {1995})}\BibitemShut {NoStop}%
\bibitem [{\citenamefont {Gilbert}\ and\ \citenamefont
  {Dorfman}(2000)}]{gilbert3}%
  \BibitemOpen
  \bibfield  {author} {\bibinfo {author} {\bibfnamefont {T.}~\bibnamefont
  {Gilbert}}\ and\ \bibinfo {author} {\bibfnamefont {J.}~\bibnamefont
  {Dorfman}},\ }\href@noop {} {\bibfield  {journal} {\bibinfo  {journal}
  {Physica A}\ }\textbf {\bibinfo {volume} {282}},\ \bibinfo {pages} {427}
  (\bibinfo {year} {2000})}\BibitemShut {NoStop}%
\bibitem [{\citenamefont {Gilbert}, \citenamefont {Dorfman},\ and\
  \citenamefont {Gaspard}(2000)}]{gilbert1}%
  \BibitemOpen
  \bibfield  {author} {\bibinfo {author} {\bibfnamefont {T.}~\bibnamefont
  {Gilbert}}, \bibinfo {author} {\bibfnamefont {J.}~\bibnamefont {Dorfman}}, \
  and\ \bibinfo {author} {\bibfnamefont {P.}~\bibnamefont {Gaspard}},\
  }\href@noop {} {\bibfield  {journal} {\bibinfo  {journal} {Phys. Rev. Lett.}\
  }\textbf {\bibinfo {volume} {85}},\ \bibinfo {pages} {1606} (\bibinfo {year}
  {2000})}\BibitemShut {NoStop}%
\bibitem [{\citenamefont {Gaspard}(1997)}]{gaspard2}%
  \BibitemOpen
  \bibfield  {author} {\bibinfo {author} {\bibfnamefont {P.}~\bibnamefont
  {Gaspard}},\ }\href@noop {} {\bibfield  {journal} {\bibinfo  {journal} {J.
  Stat. Phys.}\ }\textbf {\bibinfo {volume} {88}},\ \bibinfo {pages} {1215}
  (\bibinfo {year} {1997})}\BibitemShut {NoStop}%
\bibitem [{\citenamefont {Dorfman}, \citenamefont {Gaspard},\ and\
  \citenamefont {Gilbert}(2002)}]{dorfman2}%
  \BibitemOpen
  \bibfield  {author} {\bibinfo {author} {\bibfnamefont {J.}~\bibnamefont
  {Dorfman}}, \bibinfo {author} {\bibfnamefont {P.}~\bibnamefont {Gaspard}}, \
  and\ \bibinfo {author} {\bibfnamefont {T.}~\bibnamefont {Gilbert}},\
  }\href@noop {} {\bibfield  {journal} {\bibinfo  {journal} {Phys. Rev. E}\
  }\textbf {\bibinfo {volume} {66}},\ \bibinfo {pages} {026110} (\bibinfo
  {year} {2002})}\BibitemShut {NoStop}%
\bibitem [{\citenamefont {Barra}, \citenamefont {Gaspard},\ and\ \citenamefont
  {Gilbert}(2009{\natexlab{a}})}]{barra2}%
  \BibitemOpen
  \bibfield  {author} {\bibinfo {author} {\bibfnamefont {F.}~\bibnamefont
  {Barra}}, \bibinfo {author} {\bibfnamefont {P.}~\bibnamefont {Gaspard}}, \
  and\ \bibinfo {author} {\bibfnamefont {T.}~\bibnamefont {Gilbert}},\
  }\href@noop {} {\bibfield  {journal} {\bibinfo  {journal} {Phys. Rev. E}\
  }\textbf {\bibinfo {volume} {80}},\ \bibinfo {pages} {021126} (\bibinfo
  {year} {2009}{\natexlab{a}})}\BibitemShut {NoStop}%
\bibitem [{\citenamefont {Barra}, \citenamefont {Gaspard},\ and\ \citenamefont
  {Gilbert}(2009{\natexlab{b}})}]{barra3}%
  \BibitemOpen
  \bibfield  {author} {\bibinfo {author} {\bibfnamefont {F.}~\bibnamefont
  {Barra}}, \bibinfo {author} {\bibfnamefont {P.}~\bibnamefont {Gaspard}}, \
  and\ \bibinfo {author} {\bibfnamefont {T.}~\bibnamefont {Gilbert}},\
  }\href@noop {} {\bibfield  {journal} {\bibinfo  {journal} {Phys. Rev. E}\
  }\textbf {\bibinfo {volume} {80}},\ \bibinfo {pages} {021127} (\bibinfo
  {year} {2009}{\natexlab{b}})}\BibitemShut {NoStop}%
\bibitem [{\citenamefont {Vollmer}, \citenamefont {T{\'e}l},\ and\
  \citenamefont {M{\'a}ty{\'a}s}(2000)}]{cross1}%
  \BibitemOpen
  \bibfield  {author} {\bibinfo {author} {\bibfnamefont {J.}~\bibnamefont
  {Vollmer}}, \bibinfo {author} {\bibfnamefont {T.}~\bibnamefont {T{\'e}l}}, \
  and\ \bibinfo {author} {\bibfnamefont {L.}~\bibnamefont {M{\'a}ty{\'a}s}},\
  }\href@noop {} {\bibfield  {journal} {\bibinfo  {journal} {J. Stat. Phys.}\
  }\textbf {\bibinfo {volume} {101}},\ \bibinfo {pages} {79} (\bibinfo {year}
  {2000})}\BibitemShut {NoStop}%
\bibitem [{\citenamefont {M{\'a}ty{\'a}s}, \citenamefont {T{\'e}l},\ and\
  \citenamefont {Vollmer}(2000{\natexlab{a}})}]{cross2}%
  \BibitemOpen
  \bibfield  {author} {\bibinfo {author} {\bibfnamefont {L.}~\bibnamefont
  {M{\'a}ty{\'a}s}}, \bibinfo {author} {\bibfnamefont {T.}~\bibnamefont
  {T{\'e}l}}, \ and\ \bibinfo {author} {\bibfnamefont {J.}~\bibnamefont
  {Vollmer}},\ }\href@noop {} {\bibfield  {journal} {\bibinfo  {journal} {Phys.
  Rev. E}\ }\textbf {\bibinfo {volume} {61}},\ \bibinfo {pages} {R3295}
  (\bibinfo {year} {2000}{\natexlab{a}})}\BibitemShut {NoStop}%
\bibitem [{\citenamefont {M{\'a}ty{\'a}s}, \citenamefont {T{\'e}l},\ and\
  \citenamefont {Vollmer}(2000{\natexlab{b}})}]{cross3}%
  \BibitemOpen
  \bibfield  {author} {\bibinfo {author} {\bibfnamefont {L.}~\bibnamefont
  {M{\'a}ty{\'a}s}}, \bibinfo {author} {\bibfnamefont {T.}~\bibnamefont
  {T{\'e}l}}, \ and\ \bibinfo {author} {\bibfnamefont {J.}~\bibnamefont
  {Vollmer}},\ }\href@noop {} {\bibfield  {journal} {\bibinfo  {journal} {Phys.
  Rev. E}\ }\textbf {\bibinfo {volume} {62}},\ \bibinfo {pages} {349} (\bibinfo
  {year} {2000}{\natexlab{b}})}\BibitemShut {NoStop}%
\bibitem [{\citenamefont {Gaspard}(2005)}]{gaspard1}%
  \BibitemOpen
  \bibfield  {author} {\bibinfo {author} {\bibfnamefont {P.}~\bibnamefont
  {Gaspard}},\ }\href@noop {} {\emph {\bibinfo {title} {Chaos, scattering and
  statistical mechanics}}},\ Vol.~\bibinfo {volume} {9}\ (\bibinfo  {publisher}
  {Cambridge University Press},\ \bibinfo {address} {New York},\ \bibinfo
  {year} {2005})\BibitemShut {NoStop}%
\bibitem [{\citenamefont {Vollmer}(2002)}]{vollmers}%
  \BibitemOpen
  \bibfield  {author} {\bibinfo {author} {\bibfnamefont {J.}~\bibnamefont
  {Vollmer}},\ }\href@noop {} {\bibfield  {journal} {\bibinfo  {journal} {Phys.
  Rep.}\ }\textbf {\bibinfo {volume} {372}},\ \bibinfo {pages} {131} (\bibinfo
  {year} {2002})}\BibitemShut {NoStop}%
\bibitem [{\citenamefont {Ishida}(2013)}]{ishida2012}%
  \BibitemOpen
  \bibfield  {author} {\bibinfo {author} {\bibfnamefont {H.}~\bibnamefont
  {Ishida}},\ }\href@noop {} {\bibfield  {journal} {\bibinfo  {journal}
  {Entropy}\ }\textbf {\bibinfo {volume} {15}},\ \bibinfo {pages} {4345}
  (\bibinfo {year} {2013})}\BibitemShut {NoStop}%
\bibitem [{\citenamefont {Patankar}(1980)}]{patankar}%
  \BibitemOpen
  \bibfield  {author} {\bibinfo {author} {\bibfnamefont {S.~V.}\ \bibnamefont
  {Patankar}},\ }\href@noop {} {\emph {\bibinfo {title} {Numerical heat
  transfer and fluid flow}}}\ (\bibinfo  {publisher} {Taylor \& Francis},\
  \bibinfo {year} {1980})\BibitemShut {NoStop}%
\end{thebibliography}

\end{document}